\documentclass[envcountsame,runningheads,orivec]{llncs}
\usepackage{a4}
\usepackage{amsmath}
\usepackage{amssymb}
\usepackage{xspace}
\usepackage{graphicx}
\usepackage{pstricks}
\usepackage{colortab}
\newcommand{\IN}{\mathbb{N}}
\newcommand{\IZ}{\mathbb{Z}}
\newcommand{\IQ}{\mathbb{Q}}
\newcommand{\IR}{\mathbb{R}}

\newcommand{\calO}{\mathcal{O}}
\newcommand{\COMMENTED}[1]{}

\newcommand{\mycite}[2]{\cite[\textsc{#1}]{#2}}
\newcommand{\timesc}{\times_{\text{c}}}
\newcommand{\mydiv}{\operatorname{\tt div}}
\newcommand{\myrem}{\:\operatorname{\tt rem}\:}
\newcommand{\myand}{\operatorname{\text{\sf\&}}}
\newcommand{\bigint}{Z}
\newcommand{\medint}{Y}
\newcommand{\loglog}{\operatorname{loglog}}
\newcommand{\perm}{\operatorname{perm}}
\newcommand{\sign}{\operatorname{sgn}}
\newcommand{\Card}{\operatorname{Card}}
\newcommand{\gcdex}{\operatorname{gcdex}}
\spnewtheorem{algorithm}[theorem]{Algorithm}{\bfseries}{\itshape}
\spnewtheorem{scholium}[theorem]{Scholium}{\bfseries}{\itshape}
\newtheorem{examplef}[theorem]{Example\footnotemark}
\spnewtheorem{observation}[theorem]{Observation}{\bfseries}{\itshape}
\begin{document}
\title{On Faster Integer Calculations \\ using Non-Arithmetic Primitives\thanks{%
M.~Ziegler is supported by \textsf{DFG} project \texttt{Zi1009/1-1}}}
\titlerunning{On Faster Integer Calculations using Non-Arithmetic Primitives}
\authorrunning{K.~L\"{u}rwer-Br\"{u}ggemeier and M.~Ziegler}
\author{K.~L\"{u}rwer-Br\"{u}ggemeier and M.~Ziegler}
\institute{Heinz Nixdorf Institute,
University of Paderborn, 
33095 Germany}

\date{}
\maketitle
\def\thefootnote{\fnsymbol{footnote}}\addtocounter{footnote}{1}
\begin{abstract}
The unit cost model is both convenient and largely realistic 
for describing integer decision algorithms over $+,\times$.
Additional operations like division with 
remainder or bitwise conjunction, although equally 
supported by computing hardware, may lead to a 
considerable drop in complexity. We show a variety of
concrete problems to benefit from such \emph{non-}arithmetic
primitives by presenting and analyzing corresponding fast algorithms.
\end{abstract}
\section{Introduction} \label{s:Intro}
The Turing machine is generally accepted as the appropriate model 
for describing both the capabilities (computability) and the complexity
(bit cost) of calculations on actual digital computers; 
but it is cumbersome to handle when
developing algorithms (upper complexity bounds) as well as for proving
lower bounds, and therefore often replaced by algebraic models such
as the random access machine (RAM). The latter operates on entire integers
(as opposed to bits) and comes in various flavors depending on which
primitives it is permitted to employ: 
e.g. incrementation, addition, subtraction, multiplication, division,
integer constants, 
bitwise conjunction, shifts ``$\leftarrow,\rightarrow$'', 
indirect addressing etc.
Notice that both bitwise conjunction ``$\myand$'' 
and integer division ``$\mydiv$'' (when the numerator
is not a multiple of the denominator) 
are \emph{non-}arithmetic operations over $\IZ$
yet commonly hardware supported by digital computers 
(see Section~\ref{s:Practice} below).

The choice among these instructions heavily affect a RAM's power
in comparison to the normative Turing machine;
e.g. a decision based on polynomially many applications
of $(+,-,\times)$\footnote{but no test for inequality, 
which in the sequel we shall implicitly permit} can, 
although possibly giving rise to
exponentially long intermediate results, 
be simulated within \textsf{RP} \cite{Schoenhage};
whereas polynomially many steps over 
$(+,\times,\mydiv)$ cover already \textsf{NP} \cite{Schoenhage};
and over $(+,-,\times,\myand)$ even entire
\textsf{PSPACE} \cite{Pratt}. 
We are interested in the effect of these additional
instructions to selected problems of complexity much lower
than polynomial; specifically
for accelerating to linear and sublinear running times
as in the spirit of the following

\begin{examplef} \label{x:Example}%
\footnotetext{We thank \textsc{Riko Jacob} for
pointing us to Items~d) and e) in this example.}
\begin{enumerate}
\item[a)] 
Over $(+,-,\times,\mydiv)$,
not only primality test but even
\textsf{factorization} of a given integer $x$ 
is possible in time $\calO(\log x)$
linear in its binary length.
\item[b)]
Given $a,k\in\IN$ and some arbitrary integer $b\geq a^{2^k}$,
one can compute $a^{2^k}$ over $(+,-,\times,\mydiv)$
within $\calO(\sqrt{k})$ steps.
\item[c)]
Over $(+,-,\times,\mydiv)$ and using indirect addressing,
the greatest common divisor $\gcd(x,y)$
of given integers 
can be calculated in $\calO(\log N/\loglog N)$ steps.
where $N=\max\{x,y\}$.
\item[d)]
Over $(+,-,\myand,\leftarrow,\rightarrow)$
(but with\emph{out} indirect addressing as for Bucket Sort),
$n$ given integers $x_1,\ldots,x_n$ can be sorted
in $\calO(n)$; \\ over $(+,-,\times,\mydiv)$ this
can be achieved in $\calO(n\cdot\loglog\max_i x_i)$.
\item[e)]
\textsf{3SUM}, that is the question whether
to given integers $x_1,\ldots,x_n,y_1,\ldots,y_n,z_1,\ldots,z_n$
there exist $i,j,k$ with $x_i+y_j=z_k$, 
can be decided in $\calO(n)$ operations 
over $(+,-,\times,\myand)$.
\qed\end{enumerate}
\end{examplef}
Otherwise, 3SUM is considered `$n^2$--complete' in a
certain sense \cite{Gajentaan}.
Regarding d), describing the permutation mapping 
the input to its sorted output requires $\Omega(n\cdot\log n)$ bits.
Similarly, compare c) with the running time $\Theta(\log N)$
of the Euclidean algorithm attained on Fibonacci numbers
$x=F_n=N$, $y=F_{n-1}$.
And finally observe that in b) mere repeated squaring, i.e. 
without resorting
to integer division, yields only running time $\calO(k)$;
cf. Section~\ref{s:Power} below.
\begin{proof} 
a) See \cite{Shamir}; b) see \cite{Tiwari} or Section~\ref{s:Power} below; 
c) see \cite{B89}; d) see \cite{Kirkpatrick}, and \cite{Han}
for an account of more recent results on sorting using various
sets of operations and costs.

Claim~e) can be concluded from 
(the much more general considerations including word-length 
and non-uniform instruction cost analyses in) 
work \cite{Demaine} which, applied to our setting, 
simplifies to the following observation:
For $0\leq a_0,\ldots,a_{N-1},b_0,\ldots,b_{N-1}<2^{t-1}$, 
let $A:=\sum_{i=0}^{N-1} a_i\cdot 2^{ti}$, $B:=\sum_i b_i\cdot 2^{ti}$,
and $C:=\sum_i 2^{t-1}\cdot 2^{ti}$. Then 
\[ \forall i=0,\ldots,N-1: \; a_i\geq b_i 
\quad\Leftrightarrow\quad (A+C-B)\myand C\;=\;C \enspace . \]
In particular, subject to the above encodings,
``$\exists i: a_i=b_i$'' can be tested in constant time
over $(+,-,\myand)$.
Now such an encoding can be obtained for the
\emph{double} sequence $(x_i+y_j)_{i+nj}$ in
\emph{linear} time $\calO(n)$ over $(+,-,\times)$;
cf. e.g. our proof of Observation~\ref{o:Mamu} below.
\qed\end{proof}
%

\section{Polynomial Evaluation}
occurs ubiquitously in computer science, e.g. in connection with splines
or with Reed-Solomon codes. It is commonly performed by Horner's method
using $\calO(d)$ arithmetic operations where $d$ denotes the polynomial's
degree. While this has been proven optimal in many cases 
\mycite{Theorem~6.5}{ACT}, over (certain subrings of)
integers it is \emph{not}. 
Specifically, if the integer polynomial to be evaluated has coefficients which
are small (e.g. only 0s and 1s) compared to its degree, 
Horner's method can be slightly
accelerated:

\begin{proposition} \label{p:Ramsey}
Given $p_0,\ldots,p_{d-1}\in\IZ$
with $|p_n|\leq P\in\IN$ and $x\in\IZ$,
$\sum_{n=0}^{d-1} p_n x^n$ can be calculated using $\calO(d/\log_P d)$
operations over $(+,\times)$.
\end{proposition}
\begin{proof}
We treat the terms with negative coefficients separately
and may therefore suppose $p_n\geq0$.
For $k\in\IN$ to be chosen later, 
decompose $p$ into $\lceil d/k\rceil$ polynomials $q_i\in\IN[X]$
of degree less than $k$. Notice that, since their coefficients
belong to $\{0,1,\ldots,P-1\}$, there exist at most 
$P^k$ distinct such polynomials. Evaluate \emph{all} of them
at the given argument $x\in\IZ$: $P^k$ separate executions of
Horner's method result in a total running time of $\calO(k\cdot P^k)$;
but dynamic programming achieves the same within $\calO(P^k)$.
In a second phase, apply Horner to evaluate
$\sum_{i=0}^{\lceil d/k\rceil} q_i(x)\cdot Y^i$ at
$Y=x^k$ and obtain $p(x)$ as desired:
Together this leads to a total number
of operations of $\calO(d/k + P^k)$ which,
for $k\approx\log_P d-\log_P\log_P d$, 
becomes $\calO(d/\log_P d)$ as claimed.
\qed\end{proof}
\subsection{Throwing in Integer Division} 
The running time obtained in Proposition~\ref{p:Ramsey}
is sublinear but still dependent on the degree $d$ of
the polynomial under consideration. For fixed $p$,
\cite{Root,B92} had observed that,
surprisingly, this dependence vanishes when admitting 
integer division as operational primitive:

\begin{algorithm} \label{a:Bshouty}
Fix $X\in\IN$ and an integer polynomial with nonnegative coefficients
$p\in\IN[x]$. Then, for $\bigint\in\IN$ sufficiently large,
one can evaluate $\{0,1,\ldots,X\}\ni x\mapsto p(x)$ as follows:
\begin{enumerate}
\item[1)] Input $x\in\{0,1,\ldots,X\}$.
\item[2)] Compute $\bigint^{d+1} \;\mydiv\; (\bigint-x)$
\item[3)] multiply the result to $p(\bigint)$
\hfill\raisebox{0.0ex}[0pt][0pt]{\includegraphics[width=0.5\textwidth]{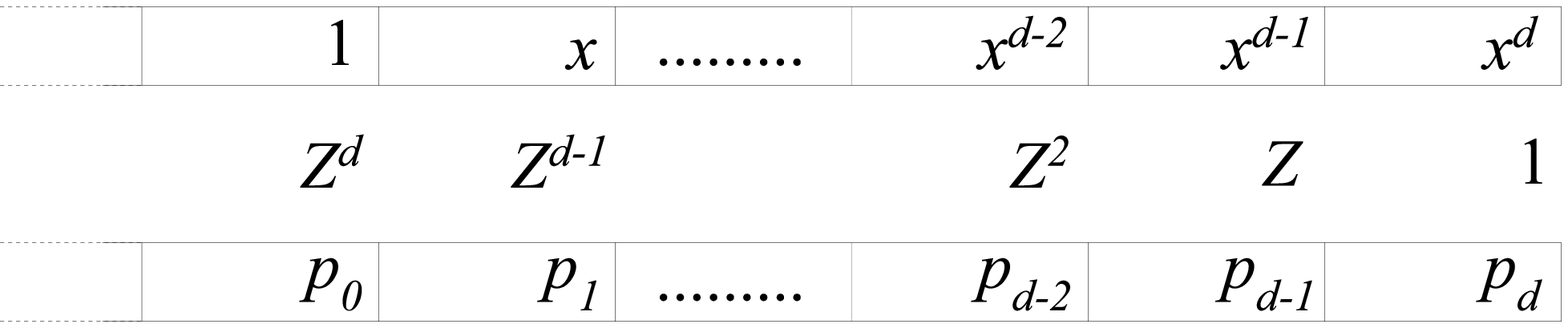}}%
\item[4)] integer-divide this in turn by $\bigint^d$ and
\item[5)] output the remainder from division by $\bigint$.
\qed\end{enumerate}
\end{algorithm}
By pre-computing and storing the integers $\bigint,\bigint^d$, and $p(\bigint)$
one arrives at

\begin{corollary} \label{c:Bshouty}
Over integer operations $\{+,-,\timesc,\mydiv\}$,
an \emph{arbitrary} fixed polynomial $p\in\IZ[x]$ can be
evaluated on an arbitrary finite domain $D\subseteq\IZ$
in constant time
\emph{independent} of $p$ and (of its degree and of) $D$.
\end{corollary}
Here, $\timesc$ denotes unary multiplication (scaling) 
of the argument by a fixed integer constant.
Indeed, evaluation of $p$ at a negative argument $-x$ 
reduces to the evaluation at positive $x$ of $p(-x)$;
and every integer polynomial is the difference of two
with nonnegative coefficients.

Concerning the correctness of Algorithm~\ref{a:Bshouty},
we repeat a proof due to \cite{B92} and obtain the following strengthening
used in Section~\ref{s:Practice}:
\begin{scholium} \label{s:Bshouty}
Let $p=\sum_{n=0}^d p_n x^n$ be of degree at most $d$ 
and norm $\|p\|_1:=|p_0|+\cdots+|p_d|\leq P$;
then every $\bigint>\max\{X^d\cdot P,(X^d+1)\cdot X\}$ \emph{is}
`sufficiently large' for Algorithm~\ref{a:Bshouty} to succeed.
\end{scholium} 
\begin{proof} It holds \ 
$\bigint^{d+1} \;\operatorname{div}\; (\bigint-x)
\;=\; \;\big\lfloor \bigint^d / (1-\tfrac{x}{\bigint}) \big\rfloor \;=\;
 \big\lfloor \bigint^d \cdot\sum\nolimits_{m=0}^\infty (x/\bigint)^m \big\rfloor \;=$
\begin{eqnarray}
&=& \lfloor \underbrace{\bigint^d + \bigint^{d-1} x +\cdots + \bigint x^{d-1} + x^d}_{\in\IN}  \nonumber
  \;\;+\;\;
\underbrace{\sum\nolimits_{m=d+1}^\infty \bigint^d (x/\bigint)^m}_{=x^{d+1}/(\bigint-x)<1}\rfloor \\[-2ex] 
&=& \bigint^d + \bigint^{d-1} x +\cdots + \bigint x^{d-1} + x^d \label{e:GeomDiv}
\end{eqnarray}
since $\bigint>(x^d+1)\cdot x$ by prerequisite. Hence the result of Step~3 equals
\[ \sum\nolimits_{n,m=0}^d  p_n\cdot x^m\cdot \bigint^{n+d-m} 
 \;=\; \sum\nolimits_{k=0}^{2d} q_k \cdot \bigint^k \]
where $0\leq q_k<\bigint$ for $k\leq d$ because $\bigint>x^d\cdot (p_0+\cdots+p_d)$.
In particular $q_d=\sum_{n=0}^d p_n \cdot x^n=p(x)$
is isolated by Steps~4 and 5.
\qed\end{proof}
%
\subsection{First Consequences}
\begin{corollary} \label{c:sequence}
Over integer operations $\{+,-,\timesc,\mydiv\}$,
every finite integer sequence $y_0,y_1,\ldots,y_N$ 
(or, more formally, the mapping $\{0,1,\ldots,N\}\ni n\mapsto y_n$)
is computable in constant time \emph{independent} of 
(the length $N$ of) the sequence!
\end{corollary}
\begin{proof}
Consider an interpolation polynomial $p\in\IQ[X]$
of degree $\leq N+1$ such that $p(n)=y_n$, $n=0,\ldots,N$.
Take $M\in\IN$ such that $M\cdot p\in\IZ[X]$.
Apply Corollary~\ref{c:Bshouty} in order
to calculate $n\mapsto M\cdot p(n)$ in constant
time, then integer-divide the result by $M$.
\qed\end{proof}
It has been shown in \cite{JMW89} that
every language $L\subseteq\IZ$ (rather than
$\IZ^*$) which can be decided 
over $\{+,-,\timesc,\mydiv\}$ at all,
can be decided in constantly many steps;
that is in time independent of the input $x\in\IZ$---\emph{but} 
of course depending on $L$. 
\begin{observation}
Every \emph{finite} language $L\subseteq\IZ$
is decidable over integer operations $\{+,-,\timesc,\mydiv\}$
within constant time \emph{in}dependent of $L$.
\end{observation}
\begin{proof}
Let $L\subseteq\{0,1,\ldots,N\}$ and apply Corollary~\ref{c:sequence}
to the characteristic sequence $(y_0,\ldots,y_N)$ of $L$,
defined by $y_n:=1$ for $n\in L$ and $y_n:=0$ for $n\not\in L$.
\qed\end{proof}
The next subsection implies the same to hold
for finite sequences $(\vec y_0,\ldots,\vec y_N)$ in $\IZ^d$
and for finite languages $L\subseteq\IZ^d$ as long as
$d$ is fixed.
\subsection{Multi-Variate Case}
We extend Algorithm~\ref{a:Bshouty} to obtain

\begin{proposition} \label{p:multivar}
Over integer operations $\{+,-,\times,\mydiv\}$,
any fixed polynomial $p\in\IZ[x_1,\ldots,x_n]$ can be
evaluated on an arbitrary finite domain $D\subseteq\IZ^n$
in time $\calO(n)$ independent of $p$ and $X$.
\end{proposition}
\begin{proof}
We devise $2^n$ separate algorithms: one for each
of the polynomials $p(\pm x_1,\pm x_2,\ldots,\pm x_n)$
to be evaluated at \emph{non-}negative argument vectors $\vec x\in\IN^n$.
Then, for a given input in $\IZ^n$, one can in time $\calO(n)$
determine which of these polynomials 
to evaluate at $(|x_1|,|x_2|,\ldots,|x_n|)$
in order to yield the aimed value $p(x)$.
Moreover decomposition of a polynomial into a part with
positive and one with negative coefficients
reduces to the case $p=\sum_{i_1,i_2,\ldots,i_n=0}^{d-1} 
a_{i_1,\ldots,i_n}\cdot x_1^{i_1}\cdots x_n^{i_n}$
with $a_{\vec\imath}\in\IN$.

As in Equation~(\ref{e:GeomDiv}) on p.\pageref{e:GeomDiv},
$\bigint^{d}\mydiv (\bigint-x)$ equals 
$\bigint^{d-1}+\bigint^{d-2}\cdot x+\cdots+\bigint\cdot x^{d-2}+x^{d-1}$
for all integers $\bigint\geq\Omega(x^d)$.
Applied to $x_2$ and $\bigint_2:=\bigint^d$,
one obtains%
\[ \big(\bigint^{d^2}\mydiv (\bigint^d-x_2)\big)\;\cdot\;
\big(\bigint^{d}\mydiv (\bigint-x_1)\big)
\quad=\quad 
\sum\nolimits_{i_1,i_2=0}^{d-1} 
\bigint^{d^2-1-(di_2+i_1)} \cdot x_2^{di_2} \cdot x_1^{i_1} \]
and inductively, using $\calO(n)$ operations from $\{+,-,\times,\mydiv\}$,
\[ \sum\nolimits_{i_1,\ldots,i_n=0}^{d-1} 
\bigint^{d^n-1-(d^{n-1}i_n+\cdots+di_2+i_1)} 
\cdot x_n^{d^{n-1}i_n} \cdots x_2^{di_2} \cdot x_1^{i_1}
\enspace . \]
Then multiply this 
counterpart to Step~2) in Algorithm~\ref{a:Bshouty}
with the constant%
\[ p(\bigint,\bigint^d,\bigint^{d^2},\ldots,\bigint^{d^{n-1}})
\quad=\quad
\sum\nolimits_{i_1,\ldots,i_n=0}^{d-1}
a_{i_1,\ldots,i_n}\cdot\bigint^{i_1+di_2+d^2i_3+\cdots+d^{n-1}i_n} \]
(cmp. Step~3) and extract the term corresponding to
$\bigint^{d^n-1}$ (Steps~4+5).
\end{proof}
\subsection{Evaluation on \emph{all} integers:
Exploiting Bitwise Conjunction} \label{s:binary}
As opposed to Horner's method,
Algorithm~\ref{a:Bshouty} and its above generalization
restricts polynomial evaluation to arguments
$x$ from an arbitrary yet \emph{finite} domain.
Indeed Scholium~\ref{s:Bshouty} derives 
from a bound $X$ on $x$ 
one on $\bigint$ to avoid
spill-overs in the $\bigint$-ary expansion of
the product of $\bigint^{d+1}\mydiv(\bigint-x)$
with $p(\bigint)$. Now $\bigint$ can of course
be chosen adaptively with respect to $x$, 
but how do we then adapt and calculate
$p(\bigint)$ accordingly? 
This becomes possible when allowing,
in addition to integer division, bitwise
conjunction as operational primitive.

\begin{proposition} \label{p:binary}
Fix $p\in\IZ[x]$ of degree $d$.
Then evaluation ~$\IZ\ni x\mapsto p(x)$~ is possible
using $\calO(\log d)$ operations over $\{+,-,\times,\mydiv,\myand\}$.
\end{proposition}
This is much faster than Horner
and asymptotically optimal.%
\begin{figure}[htb]
\centerline{\includegraphics[width=0.95\textwidth]{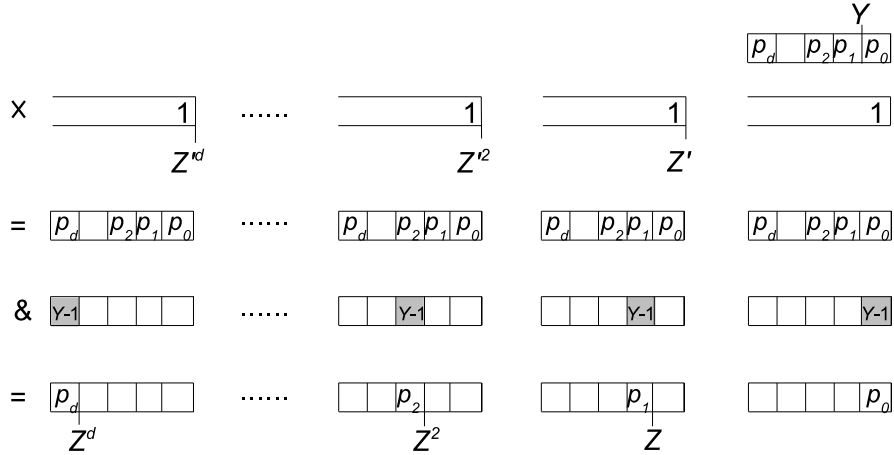}}%
\caption{\label{f:binary}Expansions of the calculations employed
in the proof of Proposition~\ref{p:binary}}
\end{figure}
\begin{proof}
As usual we may presume both $p$'s coefficients $p_0,\ldots,p_d$
and $x$ to be non-negative.
Moreover, since $p$ is fixed, one may store 
$p(\medint)$ as a constant for some sufficiently large
integer $\medint$, w.l.o.g. a power of two.
Notice that $\medint-1$ can then serve as a mask for
bitwise conjunction: for $0\leq q_n<\medint$ and
$\bigint$ a multiple of $\medint$, it holds
\[  \Big(\sum\nolimits_n q_n\cdot\bigint^n\Big)
\;\myand\; \big((\medint-1) \cdot \bigint^m\big) \;\;=\;\;
  q_m\cdot \bigint^m  \enspace ; \]
compare Figure~\ref{f:binary}.
Now given $x\in\IN$ we compute,
using repeated squaring within $\calO(\log d)$,
$\bigint':=x^{d+2}$; hence $\bigint:=\bigint'\cdot\medint$
satisfies the conditions of Scholium~\ref{s:Bshouty}.
Then, using another $\calO(\log d)$ steps,
calculate $\bigint'^{d+1}$ and, from that,
$\sum_{i=0}^d \bigint'^i=\bigint'^{d+1} \mydiv (\bigint'-1)$
as in Equation~(\ref{e:GeomDiv}).
Multiply the latter to $p(\medint)$ 
and, to the result, apply bitwise conjunction with
$\sum_{i=0}^d (\medint-1)\cdot (\bigint'\medint)^i$;
the latter can be obtained again as
$(\medint-1)\cdot\big((\bigint'^{d+1}\cdot\medint^{d+1})\mydiv 
(\bigint'\cdot\medint-1)\big)$.
Based on the mask property of $\medint-1$ mentioned above,
this yields
$\sum_{i=0}^d p_i \cdot(\bigint'\medint)^i=p(\bigint)$:
now continue as in Algorithm~\ref{a:Bshouty}.
\qed\end{proof}
A review of the above proof reveals that the $\calO(\log d)$
steps are spent for calculating 
$\bigint'=x^{d+2}$ and $\bigint^{d}$;
everything else proceeds in constant time based on
pre-computed constants like $\medint^{d}$.
Now when $x\leq\calO(2^d)$, 
$x^d$ and $\bigint^d$ are faster to obtain
starting the repeated squaring, rather than from $x$, 
from $2^d$ and $2^{d^2}$ for $\calO(\loglog x)$ steps,
respectively. Alternatively we may choose
$d$ as a power of two to invoke
Example~\ref{x:Example}b) and arrive at

\begin{scholium}
Fix $p\in\IZ[x]$ of degree $d$. Given $x\in\IZ$,
one can calculate $p(x)$ using $\calO(\loglog |x|)$
operations over $\{+,-,\times,\mydiv,\myand\}$.
If in addition some arbitrary integer $y\geq|x|^{d^2}$
is given, also running time $\calO(\sqrt{\min\{\log d,\loglog|x|\}})$ 
is feasible.
\end{scholium}
As in Proposition~\ref{p:multivar},
this extends to the multi-variate case:

\begin{theorem} \label{t:binary}
Over integer operations $\{+,-,\times,\mydiv,\myand\}$,
any fixed polynomial $p\in\IZ[x_1,\ldots,x_n]$ of maximum
degree less than $d$ 
can be evaluated in time 
\linebreak[4] $\calO(n\cdot\min\{\log d,\loglog\max_i |x_i|\})$.
\\
If, in addition to the argument $(x_1,\ldots,x_d)$, 
some integer $y\geq(\max_i |x_i|)^{d^{n+1}}$ is given,
the running time reduces to $\calO(n\cdot\sqrt{\min\{\log d,\loglog\max_i|x_i|\}})$.
\end{theorem}
\begin{proof}
According to the proof of Proposition~\ref{p:multivar},
for some integer $\bigint>\Omega(x^{d^n})$,
we need to know $(\bigint^d,\bigint^{d^2},\ldots,\bigint^{d^n})$
and $p(\bigint,\bigint^d,\ldots,\bigint^{d^{n-1}})$.
Since the latter is a \emph{uni}variate polynomial in $\bigint$ 
of degree $<d^{n+1}$, the proof of Proposition~\ref{p:binary}
shows how to obtain this value from 
$p(\medint,\medint^d,\ldots,\medint^{d^{n-1}})$
using bitwise conjunction.
Repeated squaring, either of $\max_i|x_i|$ or of $(2^d,2^{d^2},\ldots,2^{d^n})$, 
yields $(\bigint^d,\bigint^{d^2},\ldots,\bigint^{d^n})$
in time $\calO(n\cdot\min\{\log d,\loglog\max_i|x_i|\})$; 
the additional input $y$
accelerates this to $\calO(n\cdot\sqrt{\min\{\log d,\loglog\max_i|x_i|\}})$
according to Example~\ref{x:Example}b).
\qed\end{proof}

\subsection{Storing and Extracting Algebraic Numbers} \label{s:Algebraic}
When permitting not ``$\myand$'' but only $(+,-,\times\mydiv)$,
Horner's seems to remain the fastest known algorithm for evaluating an
arbitrary but fixed polynomial on entire $\IN$. Its running time
$\calO(d)$ leaves a \emph{doubly} exponential gap to the lower
bound of $\Omega(\loglog d)$ due to \mycite{Corollary~3}{LM93}.

\begin{question}  \label{q:Algebraic}
Does every (fixed) polynomial $p\in\IN[x]$
admit evaluation $x\mapsto p(x)$
on \emph{all} integers $x\in\IN$
in time $o(\deg p)$
over $(+,-,\times,\mydiv)$ ?
\end{question}
In view of the previous considerations, the answer
is positive if one can, from given $x$ within the
requested time bounds and using the operations under consideration,
obtain the number $p(\bigint)$ for some 
$\bigint>\Omega(x^{d})$ where $d>\deg p$.
To this end in turn, choose $\bigint_n:=\medint\cdot2^n$
where $\medint=2^k>\|p\|_1$ and encode the sequence
$p(\bigint_n)<\bigint_n^{d}\cdot\|p\|_1\leq2^{K+dn}$, 
where $n\in\IN$ and $K:=k\cdot (d+1)$, 
into the binary expansion of a real number like
\begin{equation} \label{e:Algebraic}
\rho_p \;:=\; \sum\nolimits_n
  p(\bigint_n)\cdot 2^{-n\cdot(K+dn)}
\enspace .
\end{equation}
Then, given $x\in\IN$, it suffices to approximate $\rho_p$
up to error $\epsilon<2^{-Kn-dn^2}$ 
for some $n\geq\Omega(d\cdot\log x)$
in order to extract\footnote{Strictly speaking, this approximation
does not permit to determine e.g. the least bit of 
$p(\bigint_n)$ due to iterated carries of less significant ones;
however this can be overcome by slightly modifying the encoding
to force the least bit to be, e.g., zero.}
the corresponding $p(\bigint_n)$.
\begin{lemma} \label{l:Algebraic}
Fix $\alpha\in\IR$ algebraic of degree $<\delta$.
Then, given $n\in\IN$, one can calculate
$u,v\in\IN$ such that $|\alpha-u/v|\leq2^{-n}$
using $\calO(\delta\cdot\log n)$
operations over $(+,-,\times)$.
\end{lemma}
\begin{proof}[Sketch]
Apply \textsf{Newton Iteration} to the 
minimal polynomial $q\in\IZ[x]$ of $\alpha$.
Since the latter is fixed, $q$, $q'$, and
an appropriate starting point for quadratic
convergence can be stored beforehand.
$\calO(\log n)$ iterations are sufficient to
attain the desired precision;
and each one amounts to evaluating
$q$ and $q'$ at cost $\calO(\delta)$
via Horner.
\qed\end{proof}
So when permitting a mild dependence of the running time on $x$
and if $\rho_p$ is algebraic of degree $o(\deg p)$,
we obtain a positive answer to Question~\ref{q:Algebraic}:
\begin{proposition}
Let $p\in\IN[x]$ be of degree $<d$ and 
suppose that $\sum_n 2^{-dn^2}$ is algebraic of degree
$<\delta$.
Then $\IN\ni x\mapsto p(x)$ can be calculated 
over $(+,-,\times,\mydiv)$ using
$\calO(\delta\cdot\loglog x)$ steps.
\end{proposition}
Unfortunately the question whether 
$\sum_n 2^{-dn^2}$ \emph{is} algebraic
(not to mention what its degree is)
constitutes a deep open problem in Number Theory
\cite[\textsc{Section~10.7.B, Example~1}, p.314]{Ribenboim2}.

We are currently pursuing a different
approach to Question~\ref{q:Algebraic} with a mild dependence
on $x$: namely by exploiting integer division in some of the 
algorithms described in \cite{Fiduccia} in combination with 
the following

\begin{observation}
Let $p\in\IQ[x]$ be of degree $<d$ and $c\in\IN$.
Then the integer sequence $p(1),p(c),p(c^2),\ldots,p(c^n),\ldots$
is linearly recurrent of degree $<d$;
that is there exist $a_0,\ldots,a_{d-1},a_d\in\IZ$
such that $p(c^{n+1})=\big(a_1\cdot p(c^n)+\cdots+a_d\cdot p(c^{n-d+1})\big)/a_0$
for all $n\in\IN$.
\end{observation}
\begin{proof}
For $k=d-1$, the $(d+1)$ polynomials $p(cx),p(x), p(x/c),\ldots,p(xc^{-k})$
all have degree $<d$ and therefore must be linearly
dependent over $\IQ$: 
$q_0\cdot p(cx)+q_1\cdot p(x)+\cdots+q_{k+1}\cdot p(xc^{-k})\equiv0$;
w.l.o.g. $q_i\in\IZ$.
Choosing $k$ minimal implies $q_0\not=0$.
\qed\end{proof}

\section{Applications to Linear Algebra}
Naive multiplication of $n\times n$ matrices takes cubic 
running time, but  \textsc{V.~Strassen}  
has set off a race for faster methods with current record
$\calO(n^\omega)$ for $\omega<2.38$ held by 
\textsc{D.~Coppersmith} and \textsc{S.~Winograd}; 
see \mycite{Section~15}{ACT} for a thorough account.
However these considerations apply to the uniform cost
model over \emph{arithmetic} operations $+,-,\times$
where division provably does not help \mycite{Theorem~7.1}{ACT};
whereas over $\IZ$ when permitting integer division as
a \emph{non-}arithmetic operation, optimal quadratic running
time can easily be attained:

\begin{observation} \label{o:Mamu}
Given $A\in\IZ^{k\times n}$ and $B\in\IZ^{n\times m}$,
one can compute $C:=A\cdot B\in\IZ^{k\times m}$ using
$\calO\big((k+m)\cdot n\big)$ operations over $\{+,-,\times,\mydiv\}$.
\end{observation}
\begin{figure}[htb]
\centerline{\includegraphics[width=0.99\textwidth]{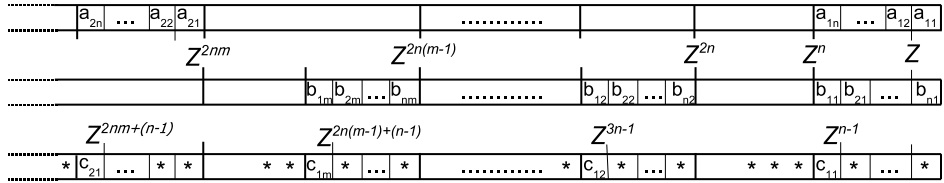}}%
\caption{\label{f:mamu}Encoding matrices $(a_{i,\ell})$ and $(b_{\ell,j})$ 
into integers $\alpha$, $\beta$ and decoding $(c_{i,j})$ from $\alpha\cdot\beta$}
\end{figure}
\begin{proof}
We want to calculate $c_{i,j}=\sum_{\ell=1}^n a_{i,\ell}\cdot b_{\ell,j}$
for $i=1,\ldots,k$ and $j=1,\ldots,m$.
W.l.o.g. $a_{i,\ell},b_{\ell,j}\geq0$; otherwise decompose.
Choose $\bigint>(\max_{i,\ell} a_{i,\ell})\cdot(\max_{\ell,j} b_{\ell,j})\cdot n$;
then compute 
\[ \alpha:=\sum_{i=1}^k\sum_{\ell=1}^n a_{i,\ell}\cdot\bigint^{(\ell-1)+2nm(i-1)}
\quad\text{ and }\quad
\beta:=\sum_{\ell=1}^n\sum_{j=1}^m b_{\ell,j}\cdot\bigint^{(n-\ell)+2n(j-1)}
\enspace . \]
As indicated in Figure~\ref{f:mamu},
the $\bigint$-adic expansion of their product $\gamma:=\alpha\cdot\beta$
contains all desired numbers $c_{i,j}$ at `position'
$\bigint^{2n(j-1)+(n-1)+2nm(i-1)}$ from which they are easily extracted
using division with remainder.
\qed\end{proof}
Observe that most of the time is spent encoding and decoding 
the input and output, respectively. However the right
factor is encoded differently from the left one; hence
binary powering yields computation of $A^k$ from
$A\in\IZ^{n\times n}$ within $\calO(n^2\cdot\log k)$
whereas a running time of $\calO(n^2+\log k)$, i.e.
by encoding and decoding only at the beginning and the
end, seems infeasible. We shall return to this topic
in Section~\ref{s:Power}.

\subsection{Determinant and Permanent} \label{s:det}
Over arithmetic operations $(+,-,\times)$, 
the asymptotic complexities of matrix multiplication 
and of determinant computation are arbitrarily close to each other
\mycite{Section~16.4}{ACT}.The same turns out to hold as well
when including integer division: not by means of reduction
but by exhibiting explicit algorithms. They reveal also
the permanent to be computable in information-theoretically 
optimal time---whereas over $(+,-,\times)$
it is \textsc{Valiant}\textsf{NP}-complete 
\mycite{Theorem~21.17}{ACT}.

\begin{proposition} \label{p:det}
Given $A\in\IZ^{n\times n}$, one can calculate
both $\det(A)$ and $\perm(A)$ within
$\calO(n^2)$ operations over $\{+,-,\times,\mydiv\}$.
\end{proposition}
Notice that, as opposed to Theorem~\ref{t:binary}, 
bitwise conjunction ``$\myand$'' is \emph{not} needed!
\begin{proof}
Both
\[
\det(A)\;=\:\sum_{\pi\in\mathcal{S}_n} \sign(\pi) a_{1,\pi(1)}\cdots a_{n,\pi(n)}
\quad\text{and}\quad
\perm(A)\;=\:\sum_{\pi\in\mathcal{S}_n} a_{1,\pi(1)}\cdots a_{n,\pi(n)}
\]
are polynomials in $n^2$ variables $x_{i-1+n(j-1)}:=a_{i,j}$
of maximum degree less than $d:=2$ with coefficients $0,\pm1$.
As in Section~\ref{s:binary} it thus suffices, in view of the proof
of Proposition~\ref{p:multivar}, to (decompose $\det$ into
one part with positive and one with negative coefficients and to)
obtain their respective values at 
$\vec x=(x_0,\ldots,x_{n^2-1}):=(\bigint',\bigint'^2,\bigint'^4,\ldots,\bigint'^{2^{n^2}})$
where $\bigint':=\bigint\cdot\medint$ for
$\bigint:=(\max_k|x_k|)^{2^{n^2}}$ and $\medint$
denotes some appropriate constant.
Now $\vec x$ can be computed in $\calO(n^2)$; and so can the quantity
\[ 
\sum_{\pi\in\mathcal{P}}
  \prod_{i=0}^{n-1} \underbrace{x_{i+n\pi'(i))}}_{=\bigint'^{i+n\pi'(i)}} \;=\;
\sum_{\pi\in\mathcal{P}} \bigint'^{\sum_{i=0}^{n-1} i+n\pi'(i)}
\;=\;\sum_{\pi\in\mathcal{P}} \bigint'^{\frac{n(n-1)}{2}+n\frac{n(n-1)}{2}}
\]
equating to
$\Card(\mathcal{P})\cdot \bigint'^{(n+1)n(n-1)/2}$
where $\pi'(i):=\pi(i+1)-1$ and $\mathcal{P}$ denotes
either $\mathcal{S}_n$ (permanent, $\Card=n!$) or one of
$\{\pi\in\mathcal{S}_n:\sign(\pi)=-1\}$ or
$\{\pi\in\mathcal{S}_n:\sign(\pi)=+1\}$ 
(determinant, $\Card=n!/2$).
\qed\end{proof}

\begin{remark}
Just recently have we been pointed out a more direct algorithm 
calculating $\IN^{n\times n}\ni A\mapsto\perm(A)$ over
$(+,-,\times,\mydiv)$ in $\calO(n^2)$ steps 
inspired by the proof of \mycite{Proposition~2.4}{Allender}.
\end{remark}

\subsection{Integer Matrix Powering: 
Exploiting GCD} \label{s:Power}
The unit cost assigned to multiplication ``$\times$'' allows to
compute high powers like $a^{2^k}$ of a given input $a$ to be
calculated by squaring $k$-times. However the presence of 
integer division \emph{and} the additional input of 
a sufficiently large but otherwise arbitrary
integer $b$ yields $a^{2^k}$
in only $\calO(\sqrt{k})$ steps; recall Example~\ref{x:Example}b).
We now generalize this to powering integer \emph{matrices}.

\begin{definition} \label{d:Power}
For $X,C\in\IZ^{d\times d}$, 
let $\gcd(C):=\gcd(c_{ij}:1\leq i,j\leq d)$ 
and $X\myrem C:=\big(x_{ij}\myrem \gcd(C)\big)$
extend the gcd and remainder from natural numbers
to integer matrices.
\\
Also write ``$X\equiv Y \pmod C$'' ~if
$\gcd(C)$ divides each entry $x_{ij}-y_{ij}$ of $X-Y$.
\end{definition}
For fixed $C$, this obviously yields an equivalence relation
on $\IZ^{d\times d}$; in fact a two-sided congruence relation\footnote{%
conversely, every two-sided ideal in $\IZ^{d\times d}$ is of this
form \mycite{Proposition~III.2.1}{Jacobson} but we won't need this fact}, 
since one easily verifies:
\begin{lemma} \label{l:Power}
\begin{enumerate}
\item[a)]
If $X\equiv Y \pmod C$, then
$S\cdot X\cdot T\equiv S\cdot Y\cdot T \pmod C$.
\item[b)]
For each $n\in\IN$ it holds $X^n \equiv Y^n \pmod{X-Y}$.
\item[c)]
$X\myrem C\equiv X \pmod C$.
\item[d)]
If $0\leq x_{ij}<\gcd(C)$ then
$X\myrem C=X$.
\end{enumerate}
\end{lemma}
Claim~b) follows from the \emph{non-}commutative binomial theorem
in connection with a).

Now apply Lemma~\ref{l:Power} to $n:=2^\ell$,
$X:=A^{2^{\ell(j-1)}}$, $Y:=B^{2^\ell}$, and $C:=Y-X$ in order to conclude
\begin{equation} \label{e:Power}
 A^{2^{\ell j}} \;=\; \big(A^{2^{\ell(j-1)}}\big)^{2^\ell}
  \;=\; B^{2^\ell} \myrem \big(B-A^{2^{\ell(j-1)}}\big)
\end{equation}
\emph{provided}
that $\gcd(B-A^{2^{\ell(j-1)}})$ is larger than the entries
of $A^{2^{\ell^2}}$.
In the case $d=1$ treated in Example~\ref{x:Example}b),
this amounts to the condition that $B=(b)$ simply be 
larger than $a^{2^{\ell^2}}$. In the general case $d>1$,
Section~\ref{s:Gcd} below
describes how to obtain matrices $B$ appropriate for our 

\begin{theorem} \label{t:Power}
Given $k\in\IN$, $A\in\IN^{d\times d}$;
furthermore given some
$B\in\IN^{d\times d}$ such that,
for all $0\leq c_{ij}<d^{2^{k}-1}\cdot (\max_{ij} a_{ij})^{2^{k}}=:r$,
it holds $\gcd(B-C)>r$;
then one can compute $A^{2^k}$ using
$\calO(d^2\cdot\sqrt{\log k})$ 
operations over $\{+,-,\times,\mydiv,\gcd\}$.
\end{theorem}
\begin{proof}
It suffices to treat the case $k=\ell^2$.
First calculate $B^{2^\ell}$ using repeated
squaring within $\calO(d^2\cdot\ell)$ 
according to Observation~\ref{o:Mamu}.
Then proceed, inductively for $j=1,\ldots,\ell$, from $A^{2^{\ell(j-1)}}$
to $A^{2^{\ell j}}$ according to Equation~(\ref{e:Power}),
again at cost $\calO(d^2\cdot\ell)$ each.
Indeed, the $m$-th power $C$ of a 
$d\times d$--matrix $A$ with entries in $\{0,1,\ldots,s\}$
has entries in $\{0,1,\ldots,d^{m-1}\cdot s^m\}$;
hence the prerequisite on $B$ shows that Lemma~\ref{l:Power}d)
applies.
\qed\end{proof}
The binary gcd operation is used to compute $\gcd(C)$ in
$\calO(d^2)$ steps and then $X\myrem C$ 
according to Definition~\ref{d:Power}. 
In fact we were surprised to realize that the above sequence
$A^{2^{kj}}$, $j=0,\ldots,k$, is obtained 
according to Equation~(\ref{e:Power}) merely by
\emph{componentwise} remainder calculations.

\subsection{Locally Lower-Bounding the GCD}  \label{s:Gcd}
Upper semi-continuity of a \emph{real} function 
$f:\IR^d\to\IR$  at  $\vec x$
means that, for arguments  $\vec u$  sufficiently close to  $\vec x$, 
the values  $f(\vec u)$  do not drop below   $f(\vec x)$  too much;
recall e.g. \mycite{Chapter~6.7}{Randolph}.
Now the greatest common divisor ($\gcd$) function is \emph{discrete}
and such topological concepts hence inapplicable in the strict sense.
Nevertheless one may say that $\gcd$ does admit 
points $\vec x$ arbitrarily close to 'approximate' upper 
semi-continuity:

\begin{lemma} \label{l:Gcd}
To $d,r,s\in\IN$ there exist $x_1,x_2,\ldots,x_d\in\IN$ such that,
for all $v_1,\ldots,v_d\in\{0,1,\ldots,s-1\}$, it holds
$\gcd(x_1+v_1,\ldots,x_d+v_d)\geq r$.
\end{lemma}
\begin{proof}
Take pairwise coprime integers $p_{\vec v}\geq r$,
$\vec v\in\{0,1,\ldots,s-1\}^d$; e.g. distinct prime numbers will do.
For $i=1,\ldots,d$ and $j=0,\ldots,s-1$ let
$u_{i,j}:=\prod_{\vec v: v_i=j} p_{\vec v}$.
Then, for fixed $i$, the numbers $u_{i,0},u_{i,1},\ldots,u_{i,s-1}$
are pairwise coprime themselves.
Hence, by the Chinese Remainder Theorem,
there exists $x_i\in\IN$ such that $u_{i,j}$ divides $x_i+j$ for 
all $j=0,1,\ldots,s-1$. In particular $p_{\vec v}$, which is
common to all $u_{i,v_i}$, divides $x_i+v_i$ for each $i=1,\ldots,d$;
and thus also divides $\gcd(x_1+v_1,\ldots,x_d+v_d)$: which therefore
must be at least as large as $p_{\vec v}\geq r$.
\qed\end{proof}
\begin{scholium} \label{sc:Gcd}
\begin{enumerate}
\item[a)]
$x_1,\ldots,x_d$ according to Lemma~\ref{l:Gcd} can be
chosen to lie between $0$ and $(r\cdot S)^{\calO(S)}$
where $S:=s^d$.
\item[b)]
It can be constructed (although not necessarily within this bound)
using $\calO(S)$ operations over $(+,-,\times,\mydiv,\gcdex)$.
\end{enumerate}
\end{scholium}
Here ``$\gcdex$'' denotes the \emph{extended} (binary) gcd function
returning, for given $x,y\in\IN$, $s,t\in\IZ$ 
(w.l.o.g. coprime) such that $\gcd(x,y)=sx+ty$.
\begin{proof}
\begin{enumerate}
\item[a)]
According to the \textsf{Prime Number Theorem},
the $k$-th prime $p_k$ has magnitude $\calO(k\cdot\log k)$
and there are at most $\pi(n)\leq\calO(n/\log n)$ primes below $n$.
Hence the first prime at least as large as $r$
has index $k_r\leq\calO(r/\log r)$; and we are interested
in bounding the product $N=p_{k_r}\cdots p_{k_r+S}$,
that is basically the quotient of \textsf{primorials}
$(r+\ell)\#/r\#$ where $r+\ell=p_{k_r+S}=r+\cal(S\cdot\log S)$.
It has been shown\footnote{We are grateful 
to our colleague \textsc{Stefan Wehmeier} for pointing
out this bound!}  \cite{Montgomery}
that $\pi(r+\ell)-\pi(r)\leq2\pi(\ell)$ holds;
that is, between $r$ and $r+\ell$ there are at most
$\calO(\ell/\log\ell)=\calO(S)$ primes; 
and each obviously
not larger than $r+\ell$. Hence
$(r+\ell)\#/r\#\leq (r+\ell)^{\calO(\ell/\log\ell)}
\leq (r\cdot\ell)^{\calO(\ell/\log\ell)}$
for $\ell=\calO(S\cdot\log S)$.
\item[b)]
Pairwise coprime integers $p_{i}\geq r$ can be found
iteratively as $p_1:=r$, $p_2:=r+1$, $p_3:=p_1\cdot p_2+1$,
and $p_{i+1}:=p_1\cdots p_i+1$. Then apply the next lemma.
\qed\end{enumerate}\end{proof}

\begin{lemma}[Chinese Remainder] \label{p:CRT}
Given integers $a_1,\ldots,a_n$ and pairwise coprime $m_1,\ldots,m_n\in\IN$,
one can calculate $x\in\IN$ with $x\equiv a_i\pmod{m_i}$ for $i=1,\ldots,n$
with $\calO(\log n\cdot\sum_{i=1}^n\log m_i)$ operations over $(+,-,\times,\mydiv)$.
\
When permitting in addition $\gcdex$ as primitive,
the running time drops down to $\calO(n)$.
\end{lemma}
\begin{proof}
Calculate $N:=m_1\cdots m_n$ and, for each $i=1,\ldots,n$,
$e_i:=t_i N/m_i$ where $1=\gcd(m_i,N/m_i)=s_im_i+t_i N/m_i$
with $s_i,t_i\in\IZ$ returned by $\gcdex$.
Then it holds 
$e_i\equiv 1\pmod{m_i}$ and $e_i\equiv 0\pmod{m_j}$ for $j\not=i$;
hence $x:=e_1\cdots e_n$ satisfies the requirements.
\\
When working over $(+,-,\times,\mydiv)$, 
the extended Euclidean algorithm computes
$\gcdex(m_i,N/m_i)$ 
within $\calO(\log N)=\calO(\sum_j\log m_j)$ steps,
for each $i=1,\ldots,n$ separately:
leading to a total running time of $\calO(n\cdot\log N)$.
In order to improve this with respect to $n$,
arrange the equations ``$x\equiv a_i\pmod{m_i}$'',
$i=1,\ldots,n$, into a binary tree: 
first compute simultaneous solutions $y_{j}$ to 
$y\equiv a_{2j}\pmod{m_{2j}}$ and 
$y\equiv a_{2j+1}\pmod{m_{2j+1}}$
for $j=1,\ldots,n/2$;
then solve adjacent quadruples as
$x\equiv y_{2j}\pmod{m_{4j}\cdot m_{4j+1}}$
and $x\equiv y_{2j+1}\pmod{m_{4j+2}\cdot m_{4j+3}}$
for $j=1,\ldots,n/4$;
and so on. The $k$-th level thus consists of solving
$n/2^k$ separate $k$-tuples of congruences involving
disjoint $k$-tuples out of $m_1,\ldots,m_n$;
that is, the corresponding extended Euclidean algorithms
incur cost $\calO(\sum_i\log m_i)$ independent of
$k=1,\ldots,\calO(\log n)$.
\qed\end{proof}
\subsection{Constructing Primes Using Integer Division} \label{s:Primes}
The (last of the)
$S$ pairwise coprime numbers $p_j\geq r$ (and thus also the integers $x_i$)
computed according to Part~b) of Scholium~\ref{sc:Gcd} are of order
$\Omega(r^{2^{S-2}})$ and thus much much larger 
than the ones asserted feasible in Part~a) by choosing $p_j$ as
prime numbers. This raises the question on the benefit of
our non-arithmetic operations for calculating primes,
i.e. for solving Problem~(b) mentioned in the beginning
of \mycite{Chapter~3}{Ribenboim} and addressed in
\textsc{Section~II} thereof.

The Sieve of Eratosthenes finds all primes up to $N$
using $\calO(N)$ operations over $(+,-)$. 
This can be accelerated \cite{Pritchard} to $\calO(N/\loglog N)$;
which is almost optimal in view of the output consisting
of $\Theta(N/\log N)$ primes according to the 
\textsf{Prime Number Theorem}. This also yields a
simple randomized way of finding a prime:

\begin{observation}
Given $N\in\IN$, guess some integer $N\leq n<2N$.
Then, with probability $\Theta(1/\log N)$, 
$n$ is a prime number: hence after $\calO(\log N)$
independent trials we have, with constant probability,
found one. Using Example~\ref{x:Example}a) to
test primality, this leads to $\calO(\log^2 N)$
expected steps over $(+,-,\times,\mydiv)$.
\end{observation}
Indeed the Bertrand-Chebyshev Theorem asserts
a prime to always exist between $N$ and $2N$.
This trivial algorithm can be slightly improved:

\begin{proposition} \label{p:primes}
Given $N\in\IN$, a randomized algorithm can, at constant
probability and within 
$\calO(\log^2 N/\loglog N)$ steps over $(+,-,\times,\mydiv)$,
obtain a prime $p\geq N$.
\end{proposition}
\begin{proof}
%
First check whether $N$ itself is prime: by 
testing whether $N$ divides $(N-1)!$
(Wilson's Theorem); using integer division,
this can be done in $\calO(\log N)$
operations over $(+,-,\times,\mydiv)$ \mycite{Section~3}{Shamir}.
From that, each adjacent factorial $(N+k)!$, $k=0,\ldots,K-1$,
is reachable in constant time: that is, after having tested
primality of $N$, corresponding checks for $N+1,N+2,\ldots,\ldots,N+K$
are basically free when $K:=\calO(\log N)$.

So now guess some $\calO(\log N)$-bit number $M\leq N$ 
and then test integers $N+M,N+M+1,\ldots,N+M+K$
in total time $\calO(\log N)$ as above.
We claim that this succeeds with probability $\Omega(\loglog N/\log N)$,
hence the claim follows by repeating independently at
random for $\calO(\log N/\loglog N)$ times.
Indeed the \textsf{Prime Number Theorem} asserts
between $N$ and $2N$ to lie $\Omega(N/\log N)$ primes;
on the other hand every interval of length $K$ between
$N$ and $2N$ contains at most $\pi(K)\leq\calO(K/\log K)$
primes \cite{Montgomery}: hence by pigeon hole, 
among these $N/K$ intervals, 
a fraction of at least $\Omega(\log K/\log N)$ must contain
at least one prime.
\qed\end{proof}
Concerning an even faster and deterministic way of constructing
primes, we

\begin{remark} \label{r:Primes}
In 1947, \textsc{W.H.~Mills} proved the existence of a real
number $\theta\approx1.3063789$ \cite{Caldwell} such that
$p_n:=\lfloor \theta^{3^n}\rfloor$, $n\in\IN$, 
yields a (sub-)sequence of primes with $p_{n+1}>p_n^3$.
It is not known whether
$\theta$ is rational; if it is,
one can straight-forwardly extract from $\theta$ 
a prime $p_n>3^n=:N$
within $\calO(n)=\calO(\log N)$ steps over $(+,-,\times,\mydiv)$.

But even if $\theta$ turns out as an \emph{algebraic} \emph{ir}rational,
then still we obtain the same time bounds! ~
Indeed, in order to compute $\lfloor\theta^{N}\rfloor$,
\[ (\theta+\epsilon)^N 
\;\;=\;\; \theta^N \;+\; \underbrace{N\cdot\epsilon\cdot\theta^{N-1} 
+ \sum\nolimits_{k=2}^N \binom{N}{k}\cdot\epsilon^k\cdot\theta^{N-k}}_{<1} \]
shows that 
it suffices to calculate a rational approximation
$\theta'$ of $\theta$ up to error $\epsilon\approx2^{-N}/N$,
according to Lemma~\ref{l:Algebraic} 
in time $\calO(\log N)$, and to then
take $\lfloor \theta'^{N}\rfloor$.
\qed\end{remark}
%

\section{Practical Relevance} \label{s:Practice}
Any real computer is of course far from able to operate 
in constant time on arbitrary large integers,
the above algorithms therefore not practical in any way.
Or are they? 
The technological progress described by \textsf{Moore's Law}
over the last decades includes an exponential increase in the width of 
processors' arithmetical-logical units (ALU). Indeed, nowadays CPUs can
commonly operate on 64 or even on 128 bits in one single instruction: 
that is, the unit cost model \emph{is} valid for surprisingly large inputs
and likely to \emph{become} valid for even larger ones.

\begin{figure}[htb]
\centerline{\large\fbox{%
\begin{tabular}{r@{\;\;}c@{\;}c@{\;}|c@{\;}c@{\;}c@{\;}|c}
\LCC & & \lightgray & & \lightgray & & \lightgray \\
$\deg(p)\;\leq$   & \;5 & \;5 & \;4 & \;4 & \;4 & \;3  \\[0.3ex] 
$\|p\|_1\;\leq$   & \;5 & \;15& \;9 & \;13 &\;23& \;56 \\[0.3ex] 
$0\leq\: x\;\leq$ & \;4 & \;3 & \;8 & \;7 & \;6 & \;21
\\[0.3ex]
$\bigint=$ & \;\texttt{0x1401} & \;\texttt{0x1000} & \;\texttt{0x9001} & \;\texttt{0x8000} & \;\texttt{0x7471} 
& \;\texttt{0x80000}\ECC
\end{tabular}}}
\caption{\label{f:Ranges}Polynomials and argument ranges
for Algorithm~\ref{a:Bshouty} to work on \textsc{x86-64} CPUs}
\end{figure}

Specifically concerning Algorithm~\ref{a:Bshouty},
it already now covers many polynomials of degree up to five
(i.e. with six free coefficients):
one can easily see the largest intermediate result
to arise from the multiplication in Step~3; which then gets
integer-divided (and thus smaller again) in Step~4.
This corresponds rather nicely to two instructions 
provided by systems like {AMD64} \mycite{Section~3.3.6}{AMD64}
and {Intel64} \mycite{Section~3.2}{Intel64}:
\texttt{mulq} multiplies two 64-bit unsigned integers to
return a full 128-bit result; while \texttt{divq} 
obtains both quotient and remainder of dividing a 
128-bit numerator by a 64-bit denominator.
So whenever, in addition to the conditions
on $\bigint$ imposed by Scholium~\ref{s:Bshouty}, 
$p(\bigint)<2^{64}$ holds,
each step of Algorithm~\ref{a:Bshouty} translates
straight-forwardly to \emph{one} \textsf{x86-64} instruction.
Figure~\ref{f:Ranges} lists some example classes of
polynomials\footnote{polynomials of higher degree $D$
can be treated as $\lceil D/(d+1)\rceil$ polynomials
of degree $\leq d$ and then applying Horner's method to $x^{d+1}$}
and argument ranges which comply with
these constraint; the shaded areas indicate that
$\bigint$ can be chosen as a power of 2 to further
replace the integer divisions in Steps~4 and 5
by a shift and a binary conjunction, respectively.
This leads to the realization indicated in the
left of Figure~\ref{f:assembler}.

\begin{figure}[htb]\hspace*{\fill}
\begin{minipage}[c]{0.40\textwidth}
\fbox{\begin{minipage}[c]{0.95\textwidth}\texttt{%
\#Input: $\bigint-x$  in \%rsi\\
\#Constants:\\
\# $p(\bigint)$ in \%rdi.\\
\# $\bigint^{d+1}$ in \%rdx:\%rax\\
\# \ (may occupy >64bit)\\
\# $\bigint-1$  in \%ebx,\\
\# $64-d\cdot\log_2(\bigint)$ in \%cl\\
divq  \%rsi\\
mulq  \%rdi\\
shld  \%cl,\%rax,\%rdx\\
andl  \%ebx,\%edx\\
\#Output: $p(x)$ in \%edx,\\
\# \%rax destroyed.
}
\end{minipage}}\end{minipage}\hfill\begin{minipage}[c]{0.40\textwidth}
\fbox{\begin{minipage}[c]{0.95\textwidth}\texttt{%
\# $x$ (byte!) in \%eax = \%ecx;\\
\# $d+1$ coeff bytes in (\%esi)\\
mulb (\%esi)\\
xorl \%ebx,\%ebx\\
movb 1(\%esi),\%bl\\
addl \%ebx,\%eax\\
mull \%ecx\\
\hspace*{1ex}:\\
mull \%ecx\\
movb $d$(\%esi),\%bl\\
addl \%ebx,\%eax\\
\# $p(x)$ in \%eax,\\
\# \%ecx \%ebx \%edx destroyed%
}
\end{minipage}}\end{minipage}\hspace*{\fill}
\caption{\label{f:assembler}\textsc{x86-64} GNU assembler realization
of Algorithm~\ref{a:Bshouty} and of Horner's method}
\end{figure}

In comparison with Horner's Method depicted to the right,
this amounts to essentially the elimination of $d-1$ (out of $d$)
multiplications at the expense of one division---in a sense a counter-part to
the converse direction, taken e.g. in \cite{Granlund}, of replacing
integer divisions by multiplications.

Now an actual performance prediction, and even a meaningful
experimental evaluation, is difficult in the age
of caching hierarchies and speculative execution. For instance
(e.g. traditional) 32-bit applications may leave large parts 
of a modern superscalar CPU's 64-bit ALU essentially
idle, in which case the left part of 
Figure~\ref{f:assembler} as a separate (hyper-)thread
can execute basically for free.

However even shorter than both Horner's and Bshouty's Algorithm 
for the evaluation of a fixed polynomial $p$ is
one (!) simple lookup in a pre-computed table storing 
$p$'s values for $x=0,1,\ldots,X$. On the other hand
when there are \emph{many} polynomials to be evaluated,
the tables required in this approach may become pretty
large; e.g. in case of $d=3$, $X=21$, and $\|p\|_1\leq56$
(right-most column in Figure~\ref{f:Ranges}), the
values of $p(x)$ reach up to $X^d\cdot\|p\|_1$, hence
do not fit into 16 bit and thus occupy a total
of $(X+1)\times 4$ bytes for each of the 
$\binom{\|p\|_1+d+1}{d+1}=487,635$ possible polynomials $p$:
far too much to be held in cache and thus prone to
considerably stall a nowadays computer; whereas 
the $487,635$ possible 64-bit values $p(\bigint)$
do fit nicely into the 4MB L2-cache of modern CPUs, 
the corresponding four byte coefficients per polynomial
(cf. right part of Figure~\ref{f:assembler})
even into 2MB. One may therefore regard Algorithm~\ref{a:Bshouty}
as a compromise between table-lookup and Horner's Method.
\section{Conclusion}
We presented algorithms which, using integer division and
related non-arithmetic operations like bitwise conjunction
or greatest common divisor, accelerate polynomial
evaluation, linear algebra, and number-theoretic calculations to 
optimal running times. Several solutions would depend on
deep open number-theoretical hypotheses, showing that
corresponding lower bounds are probably quite difficult to obtain.
Other problem turned out as solvable
surprisingly fast (and actually \emph{beating} information-theoretical
lower bounds) when providing some more or less generic
integers as additional input.

On the other hand, these large numbers would suffice to
be of size `only' doubly exponential---and thus quickly
computable when permitting
leftshifts $\leftarrow:y\mapsto 2^y$
or, more generally, exponentiation 
$\IN\times\IN\ni(x,y)\mapsto x^y$ as primitive
at unit cost.
In view of the hierarchy ``addition, multiplication, exponentiation'',
it seems interesting to gauge the benefit of 
level $\ell$ of Ackermann's function $A(\ell,\cdot)$
to seemingly unrelated natural problems over integers.


\end{document}